\documentclass[journal]{IEEEtran}

\usepackage[numbers,sort&compress,square]{natbib}   
\usepackage{graphicx}
\usepackage{subfigure}
\usepackage{amsmath,amssymb}
\usepackage{amsfonts}
\usepackage{array}
\usepackage{amsthm}
\usepackage{mathrsfs}
\usepackage{setspace}
\usepackage{color}
\usepackage[ruled,linesnumbered]{algorithm2e}
\usepackage{soul}
\usepackage{balance}

\newtheorem{theorem}{Theorem}

\newtheorem{lemma}[theorem]{Lemma}

\theoremstyle{definition}

\begin{document}

	\title{Optimizing the Trade-off Between Throughput and PAoI Outage Exponents}

\author{\IEEEauthorblockN{Tai-Chun Yeh and Yu-Pin Hsu\\} 
\IEEEauthorblockA{Department of Communication Engineering \\National Taipei University}\\
	\IEEEauthorblockA{a88700337@gmail.com, yupinhsu@mail.ntpu.edu.tw}
	\thanks{The work was supported by National Science and Technology Council, Taiwan, under the grant number 110-2221-E-305-008-MY3.}
%

}

\maketitle
\thispagestyle{empty}
\pagestyle{empty}

\begin{abstract}
This paper investigates the trade-off between throughput and peak age of information (PAoI) outage probability in a multi-sensor information collection system. Each sensor monitors a physical process, periodically samples its status, and transmits the updates to a central access point over a shared radio resource. The trade-off arises from the interplay between  each sensor's sampling frequency and the allocation of the shared resource. To optimize this trade-off, we formulate a joint optimization problem for each sensor's sampling delay and resource allocation, aiming to minimize a weighted sum of sampling delay costs (representing a weighted sum of throughput) while satisfying PAoI outage probability exponent constraints. We derive an optimal solution and particularly propose a closed-form approximation for large-scale systems. This approximation provides an explicit expression for an approximately optimal trade-off,  laying a foundation for designing resource-constrained systems in applications that demand frequent updates and also stringent statistical timeliness guarantees.

\end{abstract}


\section{Introduction}

In recent years, the rapid advancement of smart systems has significantly increased the demand for timely information to ensure their efficient and effective operation. For instance, in a smart transportation system \cite{qin2021distributed}, vehicles are equipped with sensors (e.g., GPS or radar) that continuously monitor their surroundings. These sensors generate updates about various physical processes and transmit the updates to a central controller. By aggregating these updates, the central controller gains a real-time understanding of the dynamic environment, enabling smart decision-making. To ensure safety and reliability, the information maintained by the central controller must remain as timely as possible.

To quantify the timeliness of information, \cite{kaul2012real} introduced the concept of the age of information (AoI), which measures the time elapsed since the generation of the most recent update. A lower AoI indicates that the received update is closer to the current state. Building on this, \cite{yates2018age} proposed the peak age of information (PAoI), which represents the maximum AoI observed just before a new update is received. PAoI is particularly valuable for systems where worst-case timeliness or adherence to specific age thresholds is critical.

Extensive research has been conducted on analyzing and minimizing average AoI/PAoI in diverse system settings. For instance, \cite{costa2016age} derived analytical expressions for the average AoI and PAoI in single-source systems, while \cite{yates2018age} extended this analysis to multi-source systems. Moreover, numerous system design strategies have been proposed to minimize these metrics, e.g., scheduling \cite{kadota2018scheduling}, resource allocation \cite{park2020centralized}, and sampling  \cite{ornee2019sampling}. Comprehensive surveys of these efforts can be found in \cite{yates2021age, kahraman2023age}. However, while minimizing the average AoI/PAoI typically enhances timeliness, it does not ensure a desired AoI with high probability. For systems with strict timeliness requirements, statistical AoI/PAoI guarantees are crucial. This has motivated research into the probabilistic characterization of AoI/PAoI, e.g., \cite{seo2019outage, hu2021status, zhong2023stochastic}.

In addition to timeliness, the total number of status updates (usually referred to as throughput) is a critical factor influencing the performance of smart systems. For example, in a smart transportation system, frequent sampling of a vehicle’s location allows the central controller to predict its future position with greater accuracy. However, overly frequent sampling can lead to network congestion, causing higher queueing delay and ultimately stale updates. To address this, numerous studies have examined the trade-off between throughput and \textit{average} AoI/PAoI, e.g., \cite{mankar2021throughput, bhat2020throughput}. Despite these efforts, the  relationship between throughput and \textit{statistical} AoI/PAoI  remains unclear. This gap is particularly relevant for systems that demand frequent updates while adhering to stringent statistical AoI/PAoI guarantees, such as those employed in safety-critical applications.

To address this gap, this paper investigates the trade-off between throughput and PAoI outage probability in an information status collection system. The system consists of multiple local monitors that observe physical processes and periodically generate status updates, which are transmitted to a central controller over a shared communication resource. Since PAoI is influenced by both sampling delay and transmission delay (the latter determined by resource allocation), optimizing only one of these factors is insufficient to achieve the optimal trade-off. 

The primary contribution of this paper is the joint optimal design of  sampling delay and resource allocation. The objective is to minimize a total sampling delay cost while satisfying PAoI outage probability exponent constraints, providing a framework to analyze the trade-off between  throughput and  PAoI outage. Our main result reveals that an approximately optimal  resource allocation consists of the minimum required allocation to satisfy the PAoI outage constraints, supplemented by a proportional distribution of the remaining resource. Moreover, through the approximation, we derive a closed-form expression that captures the relationship between each sensor's sampling delay (indicative of each sensor's throughput) and the PAoI outage probability exponents under the approximately optimal resource allocation. These findings are critical for designing systems that require frequent and timely information updates.

\section{System overview}\label{section:system}

This section introduces a network model in Section~\ref{subsection:model} and formulates a problem addressing the trade-off between sampling delay and PAoI outage  in Section~\ref{subsection:problem}.

\subsection{Network model}\label{subsection:model}
\begin{figure}
	\centering
	\includegraphics[width=.4\textwidth]{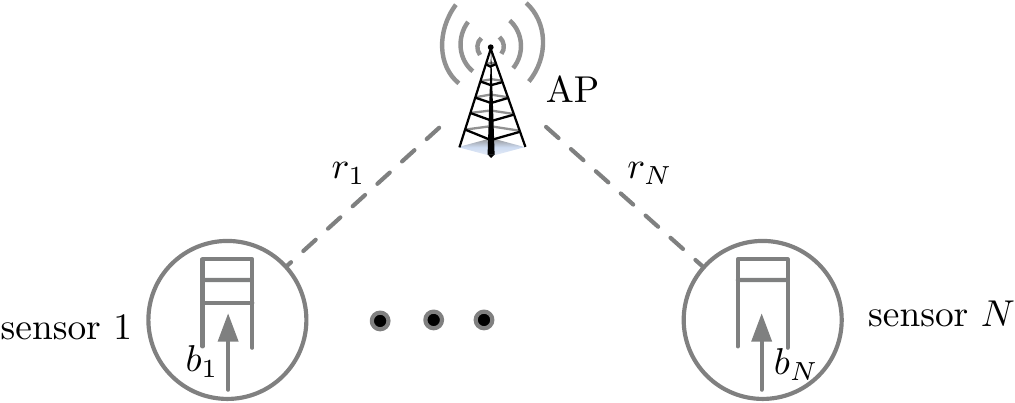}
	\caption{An information collection system.}
	\label{fig:model}
\end{figure}

As illustrated in Fig.~\ref{fig:model}, we consider an information status collection system comprising $N$ sensors and an access point (AP). Each sensor monitors a physical process and  updates the AP with the latest status of the process.  Specifically, each sensor samples the physical process periodically. Let $b_i$ represent the sampling delay for sensor $i$, and define $\mathbf{b} = (b_1, \cdots, b_N)$. Additionally, let $S_{i,j}$ denote the time when sensor $i$ generates  sample $j$.

The generated samples are transmitted to the AP as packets using a shared radio resource. Let $r_i$ represent the fraction of the total resource allocated to sensor $i$, satisfying $\sum_{i=1}^N r_i \leq 1$. Define  $\mathbf{r} = (r_1, \cdots, r_N)$. Given a resource allocation $\mathbf{r}$, the transmission time of a packet is assumed to be independent of other packets and follows the same statistical distribution for packets from the same sensor. Let $T_i$ denote the transmission time of a packet from sensor $i$, where the distribution of $T_i$ depends on the allocated resource $r_i$ only. Each sensor maintains a queue to store packets awaiting transmission, and the packets in the queue are transmitted in a first-come first-served manner.

To evaluate the timeliness of information at the AP, we use the peak age of information (PAoI) metric proposed in \cite{yates2018age}. PAoI quantifies the maximum staleness of information just before an update is received at the AP. Formally, let $D_{i,j}$ represent the time when  packet $j$ from sensor $i$ is delivered to the AP.  PAoI for sensor $i$, denoted by $A_{i,j}$, upon receiving packet~$j$ is given by $A_{i,j} = D_{i,j} - S_{i,j-1}$. Unlike most existing studies that focused on average AoI/PAoI, this work examines  outage probability of PAoI, which is crucial for applications requiring strict statistical guarantees. Specifically, for sensor~$i$, the AP enforces a constraint $\theta_i > 0$ such that
\[
\lim_{x \to \infty} \frac{\ln\mathbb{P}[A_{i,\infty} \geq x]}{x} \leq -\theta_i,
\]
ensuring that the outage probability $\mathbb{P}[A_{i,\infty} \geq x]$ in the steady state (as $j \to \infty$)  decays exponentially with a rate of at least $\theta_i$, as $x$ increases.  Here, $\theta_i$ is referred to as the outage exponent. The higher the value of $\theta_i$, the more stringent the probabilistic PAoI requirement. Let $\boldsymbol{\theta} = (\theta_1, \cdots, \theta_N)$.


\subsection{Problem formulation}\label{subsection:problem}

To provide more timely information to the AP, a shorter sampling delay is desired, which, however, increases queueing delay and consequently raises PAoI. To study this trade-off, let $C_i$ denote the cost of delaying sampling for sensor $i$ by one unit of time, and define $\mathbf{C} = (C_1, \cdots, C_N)$. The goal is to minimize the total sampling delay cost $\sum_{i=1}^N C_i b_i$  (reflecting a high weighted sum of throughput since a lower $b_i$ corresponds to a higher sampling frequency) while ensuring the outage exponents  for all sensors are satisfied. 

Since PAoI depends on both  sampling delay $\mathbf{b}$ and  resource allocation $\mathbf{r}$, the problem requires jointly designing $\mathbf{b}$ and $\mathbf{r}$. Precisely, given outage exponent  $\boldsymbol{\theta}$ and sampling delay cost $\mathbf{C}$, we formulate the following optimization problem:
\begin{subequations}\label{eq:opt-original}
\begin{align}
\min_{\mathbf{b}, \mathbf{r}}  &\,\,\, \sum_{i=1}^N C_i b_i \label{eq:opt-original:objective} \\
\textrm{s.t.}  &\,\,\, \lim_{x \to \infty} \frac{\ln\mathbb{P}[A_{i,\infty} \geq x]}{x} \leq -\theta_i, \text{ for all $i = 1, \cdots, N$;} \label{eq:opt-original:constraint1} \\
&\,\,\, \sum_{i=1}^N r_i \leq 1. \label{eq:opt-original:constraint2}
\end{align}
\end{subequations}
Note that this formulation captures the trade-off between the total sampling delay cost and the statistical PAoI guarantees in a heterogeneous network, where different sensors may have distinct sampling delay costs or transmission times. 

\section{Optimal design for sampling delay}
In this section, we characterize the optimality structure of the sampling delay design. To this end, let $\Lambda_i(\theta) = \ln \mathbb{E}[e^{\theta T_i}]$ represent the log moment generating function (LMGF) of transmission time $T_i$. From \cite[Lemma 1]{seo2019outage}, if the condition $\frac{\Lambda_i(\theta_i)}{\theta_i} < b_i$ (with strict inequality) is satisfied, then the asymptotic behavior of PAoI satisfies $\lim_{x \to \infty} \frac{\ln\mathbb{P}[A_{i,\infty} \geq x]}{x} \leq -\theta_i$, ensuring that the outage exponent $\theta_i$ is achievable. Building on this result, we further establish a stronger characterization in the following lemma. Specifically, we show that the condition $\frac{\Lambda_i(\theta_i)}{\theta_i} \leq b_i$ is not only sufficient but also necessary for satisfying the outage constraint in Eq.~(\ref{eq:opt-original:constraint1}). 

\begin{lemma}\label{lemma:necessary-condition}
The condition $\lim_{x \to \infty} \frac{\ln\mathbb{P}[A_{i,\infty} \geq x]}{x} \leq -\theta_i$ holds if and only if $\frac{\Lambda_i(\theta_i)}{\theta_i} \leq b_i$.
\end{lemma}

\begin{proof}[Proof (sketch)]
 We employ large deviation theory \cite{srikant2014communication} to show that the AoI outage probability exponent can be expressed by $\lim_{x \to \infty} \frac{\ln\mathbb{P}[A_{i,\infty} \geq x]}{x}=-\inf_{t>0} \frac{1}{t} I_i(t + b)$, where $I_i(x) = \sup_{\gamma \in \mathbf{R}} \{\gamma x - \Lambda_i(\gamma)\}$ is the rate function corresponding to  transmission time $T_i$. Then, leveraging the relation between the rate function $I_i(\cdot)$ and the LMGF $\Lambda_i(\cdot)$, we can establish the equivalence. The detailed proof is given in Appendix~\ref{appendix:lemma:necessary-condition}.
\end{proof}

Lemma~\ref{lemma:necessary-condition} is analogous to the effective bandwidth theory (e.g., see \cite{de1995effective}) in queueing systems; however, their results cannot be directly applied here, as the evolution of PAoI differs from that of queues. Moreover, from Lemma~\ref{lemma:necessary-condition}, the constraint in Eq.~(\ref{eq:opt-original:constraint1}) can be equivalently replaced by $\frac{\Lambda_i(\theta_i)}{\theta_i} \leq b_i$. Using this substitution, the optimization problem in Eq.~(\ref{eq:opt-original}) can be reformulated as the following equivalent problem:
\begin{subequations}\label{eq:opt-transform}
\begin{align}
\min_{\mathbf{b}, \mathbf{r}} \quad &  \sum_{i=1}^{N} C_i b_i \label{eq:opt-transform:objective}\\
\textrm{s.t.} \quad &\frac{\Lambda_i(\theta_i)}{\theta_i} \leq b_i,  \text{ for all } i = 1, \cdots, N; \label{eq:opt-transform:constraint1} \\
& \sum_{i=1}^N r_i \leq 1. \label{eq:opt-transform:constraint2}
\end{align}
\end{subequations}

Recall that the LMGF $\Lambda_i(\theta_i)$ of  transmission time $T_i$ depends solely on  resource allocation and is independent of  sampling delay. Therefore, from Eq.~(\ref{eq:opt-transform:constraint1}), for a given resource allocation $\mathbf{r}$, the  sampling delay for sensor $i$ that minimizes the total sampling delay cost in Eq.~(\ref{eq:opt-transform:objective}) can be determined as  $b_i = \frac{\Lambda_i(\theta_i)}{\theta_i}$.
In the next section, we will derive an optimal resource allocation, denoted by $\mathbf{r}^*$. Once $\mathbf{r}^*$ is determined, the optimal sampling delay for sensor $i$, denoted by $b_i^*$, can be computed as
\begin{align}
b^*_i = \frac{\Lambda_i^*(\theta_i)}{\theta_i}, \label{eq:optimal-b}
\end{align}
where $\Lambda_i^*(\cdot)$ represents the LMGF of $T_i$ under the optimal resource allocation $r^*_i$.

\section{Optimal design for resource allocation under a specific transmission distribution}

In this section, we determine the optimal resource allocation for a specific transmission time distribution. We model  transmission time $T_i$ as an exponential random variable with mean $\mu_i r_i$, where $\mu_i$ represents the mean transmission rate per unit of resource. This model implies that the transmission capacity of a sensor scales proportionally with the amount of resource allocated to it. 

For exponentially distributed transmission times, the LMGF $\Lambda_i(\theta_i)$ can be computed as
$\Lambda_i(\theta_i) = \ln\frac{\mu_i r_i}{\mu_i r_i - \theta_i}$.
Given the relationship between $r^*_i$ and $b^*_i$ derived in Eq.~(\ref{eq:optimal-b}), we substitute $b_i$ in Eq.~(\ref{eq:opt-transform:objective}) with $b_i = \frac{1}{\theta_i}\ln\frac{\mu_i r_i}{\mu_i r_i - \theta_i}$.
This substitution transforms the optimization problem in Eq.~(\ref{eq:opt-transform}) into the following equivalent form for determining the optimal resource allocation $\mathbf{r}$:
\begin{subequations}\label{eq:opt-final}
\begin{align}
\min_{\mathbf{r}} \quad & \sum_{i=1}^{N} \frac{C_i}{\theta_i} \cdot \ln\frac{\mu_i r_i}{\mu_i r_i - \theta_i} \label{eq:opt-final:objective} \\
\textrm{s.t.} \quad & \sum_{i=1}^N r_i \leq 1. \label{eq:opt-final:constraint}
\end{align}
\end{subequations}

\subsection{Feasibility region of outage exponent $\boldsymbol{\theta}$}
We note that the objective function in Eq.~(\ref{eq:opt-final:objective}) is feasible only when $\mu_i r_i > \theta_i$. This implies that if sufficient resource $r_i$ cannot be allocated to sensor $i$ to ensure its transmission rate $\mu_i r_i$ exceeds its outage exponent $\theta_i$, the outage exponent guarantee cannot be achieved (as stated in Lemma~\ref{lemma:necessary-condition}). Moreover, if no resource allocation $\mathbf{r}$ satisfies $\mu_i r_i > \theta_i$ for all $i$ and $\sum_{i=1}^N r_i \leq 1$, the optimization problem in Eq.~(\ref{eq:opt-final})  have no solution.

Thus, we define outage exponent  $\boldsymbol{\theta}$ as \textit{feasible} if and only if there exists a resource allocation $\mathbf{r}$ such that $\mu_i r_i > \theta_i$ for all $i$ and also $\sum_{i=1}^N r_i \leq 1$. The following lemma provides a necessary and sufficient condition for the feasibility of $\boldsymbol{\theta}$.

\begin{lemma}
A set $\boldsymbol{\theta}$ of outage exponents is feasible if and only if $\sum_{i=1}^N \frac{\theta_i}{\mu_i} < 1$.
\end{lemma}

\begin{proof}
First, assume that $\sum_{i=1}^N \frac{\theta_i}{\mu_i} < 1$. Then, there exists an $\epsilon > 0$ such that $\epsilon \leq \frac{1 - \sum_{i=1}^N \frac{\theta_i}{\mu_i}}{N}$. Consider a resource allocation for sensor $i$ as $r_i = \frac{\theta_i}{\mu_i} + \epsilon$. With this allocation, $\mu_i r_i > \theta_i$ for all $i$. Additionally, we have
\begin{align*}
\sum_{i=1}^N r_i = \sum_{i=1}^N \left(\frac{\theta_i}{\mu_i} + \epsilon\right) &= \sum_{i=1}^N \frac{\theta_i}{\mu_i} + N\epsilon \\
&\leq \sum_{i=1}^N \frac{\theta_i}{\mu_i} + \left(1 - \sum_{i=1}^N \frac{\theta_i}{\mu_i}\right) = 1.	
\end{align*}
Thus, the outage exponent  $\boldsymbol{\theta}$ is feasible.

Next, assume that $\sum_{i=1}^N \frac{\theta_i}{\mu_i} \geq 1$, and suppose $\mu_i r_i > \theta_i$ for all $i$. Then, we have
\begin{align*}
\sum_{i=1}^N r_i > \sum_{i=1}^N \frac{\theta_i}{\mu_i} \geq 1,	
\end{align*}
which contradicts the constraint $\sum_{i=1}^N r_i \leq 1$. Therefore, by contraposition, $\sum_{i=1}^N \frac{\theta_i}{\mu_i} \geq 1$ implies that the outage exponent $\boldsymbol{\theta}$ is not feasible.
\end{proof}

As a result, the set $\{\boldsymbol{\theta}: \sum_{i=1}^N \frac{\theta_i}{\mu_i} < 1\}$ defines the region of all feasible outage exponents. This feasibility region is analogous to the capacity region \cite{srikant2014communication} in queueing systems. In this paper, we focus on scenarios where the requested outage exponent  $\boldsymbol{\theta}$ lies within the feasibility region.

\subsection{Optimal solution for resource allocation}

The optimal resource allocation  can be determined by solving the optimization problem in Eq.~(\ref{eq:opt-final}). Since the problem in Eq.~(\ref{eq:opt-final}) is a convex optimization problem and satisfies the Slater condition for a feasible $\boldsymbol{\theta}$, it can be  solved using the Karush-Kuhn-Tucker (KKT) conditions \cite{srikant2014communication}. 

The Lagrangian for this problem is given by
\[
\mathcal{L}(\mathbf{r}, \lambda) = \sum_{i=1}^{N}\left( \frac{C_i}{\theta_i} \cdot \ln \frac{\mu_i r_i}{\mu_i r_i - \theta_i} \right)+ \lambda\left(\sum_{i=1}^{N} r_i - 1\right),
\]
where $\lambda > 0$ is a Lagrange multiplier. Differentiating the Lagrangian with respect to $r_i$, we obtain
\begin{align*}
\frac{\partial \mathcal{L}(\mathbf{r}, \lambda)}{\partial r_i} &= \frac{C_i}{\theta_i} \cdot \left( \frac{\mu_i r_i - \theta_i}{\mu_i r_i} \cdot \frac{\mu_i \left(\mu_i r_i - \theta_i\right) - \mu_i^2 r_i}{\left(\mu_i r_i - \theta_i\right)^2} \right) + \lambda \\
&=  \frac{-C_i}{r_i \left(\mu_i r_i - \theta_i\right)} + \lambda.
\end{align*}
Setting the derivative to zero, we can obtain $\mu_i \lambda r_i^2 - \theta_i \lambda r_i - C_i = 0$. Solving this quadratic equation, the optimal resource allocation for sensor $i$ is
\begin{align}
r_i^* = \frac{\theta_i \lambda + \sqrt{\left(\theta_i \lambda\right)^2 + 4 C_i \mu_i \lambda}}{2 \mu_i\lambda } 
= \frac{\theta_i}{2\mu_i} + \frac{\theta_i}{2\mu_i} \left(1 + \frac{4 C_i \mu_i}{\theta_i^2 \lambda}\right)^{\frac{1}{2}}. \label{eq:optimal_r}
\end{align}
The Lagrange multiplier $\lambda$ can be determined by enforcing the constraint $\sum_{i=1}^N r_i^* = 1$, leading to
\begin{align}
\sum_{i=1}^N \frac{\theta_i}{2\mu_i} + \frac{\theta_i}{2\mu_i} \left(1 + \frac{4 C_i \mu_i}{\theta_i^2 \lambda}\right)^{\frac{1}{2}} = 1. \label{eq:solve-lambda}
\end{align}

Due to the square root terms in Eq.~(\ref{eq:solve-lambda}), there is no closed-form solution for $\lambda$. However, this equation can be solved using numerical methods, such as Newton's method \cite{hildebrand1987introduction}.

\section{Approximate Solution}

To gain insights into the optimal resource allocation in Eq.~(\ref{eq:optimal_r}) and also characterize the optimal trade-off between throughput and PAoI outage exponents, this section presents an approximation. When the number $N$ of sensors  is large, $\lambda$ becomes large, as indicated by Eq.~(\ref{eq:solve-lambda}). Using the approximation $(1 + x)^{\frac{1}{2}} \approx 1 + \frac{x}{2}$ when $x$ is small, we can approximate Eq.~(\ref{eq:optimal_r}) as
\begin{align}
r^*_i \approx \frac{\theta_i}{2\mu_i} + \frac{\theta_i}{2\mu_i} \left(1 + \frac{1}{2} \cdot \frac{4 C_i \mu_i}{\theta_i^2 \lambda}\right) = \frac{\theta_i}{\mu_i} + \frac{C_i}{\theta_i} \frac{1}{\lambda}. \label{eq:approximate-r}
\end{align}

According to Eq.~(\ref{eq:approximate-r}), we can interpret $\theta_i$  as the required transmission rate for sensor $i$ to satisfy its outage exponent. Thus, the first term $\frac{\theta_i}{\mu_i}$ in Eq.~(\ref{eq:approximate-r}) corresponds to the minimum resource required for sensor $i$ to meet its outage exponent. The AP  reserves $\sum_{i=1}^N \frac{\theta_i}{\mu_i}$ resource to satisfy all outage exponents. The remaining resource, $1 - \sum_{i=1}^N \frac{\theta_i}{\mu_i}$, is distributed to minimize the total sampling delay cost.

The second term in Eq.~(\ref{eq:approximate-r}) shows that the remaining resource is allocated proportional to $\frac{C_i}{\theta_i}$. This allocation can be explained as follows: From Eq.~(\ref{eq:optimal-b}), the sampling delay $b_i^*$ is given by $b_i^* = \frac{1}{\theta_i} \ln\frac{\mu_i r_i^*}{\mu_i r_i^* - \theta_i}$.
Allocating more resource $r_i^*$ reduces $b_i^*$,  lowering the total sampling delay cost. Also, a higher $\theta_i$ increases $b_i^*$, as a stricter outage exponent requires less frequent sampling to avoid excessive queueing delay. Thus, the remaining resource allocation follows the principles as follows:
\begin{itemize}
    \item A sensor~$i$ with a high sampling delay cost $C_i$ is allocated more remaining resource to reduce its sampling delay and offset the high unit cost.
    \item A sensor~$i$ with a stringent outage exponent $\theta_i$ is allocated fewer remaining resource, resulting in a longer sampling delay to meet the stricter constraint.
\end{itemize}

Next, using the condition $\sum_{i=1}^N r_i^* = 1$, the Lagrange multiplier $\lambda$ in Eq.~(\ref{eq:approximate-r}) can be derived, and the approximately optimal resource allocation becomes
\begin{align}
r^*_i \approx \frac{\theta_i}{\mu_i} + \frac{C_i}{\theta_i} \frac{1 - \sum_{i=1}^N \frac{\theta_i}{\mu_i}}{\sum_{i=1}^N \frac{C_i}{\theta_i}}. \label{eq:approximate-r-final}	
\end{align}
Substituting this allocation into Eq.~(\ref{eq:optimal-b}), the approximately optimal sampling delay becomes
\begin{align}
b^*_i \approx \frac{1}{\theta_i} \ln \left(1 + \frac{\theta_i^2}{C_i \mu_i}  \frac{\sum_{i=1}^N \frac{C_i}{\theta_i}}{1 - \sum_{i=1}^N \frac{\theta_i}{\mu_i}}\right). \label{eq:approxiamte-b-final}	
\end{align}

Our approximation in Eqs.~(\ref{eq:approximate-r-final}) and (\ref{eq:approxiamte-b-final}) provides a computationally efficient approach to determine the resource allocation and the sampling delay, avoiding the need for complex numerical methods. Moreover, Eq.~(\ref{eq:approxiamte-b-final}) characterizes an approximately optimal trade-off between sampling frequency and outage exponents across various system parameters in a heterogeneous network. The numerical analysis of this trade-off will be discussed in the next section.


\section{Numerical results}
\begin{figure}
\centering
	\includegraphics[width=.4\textwidth]{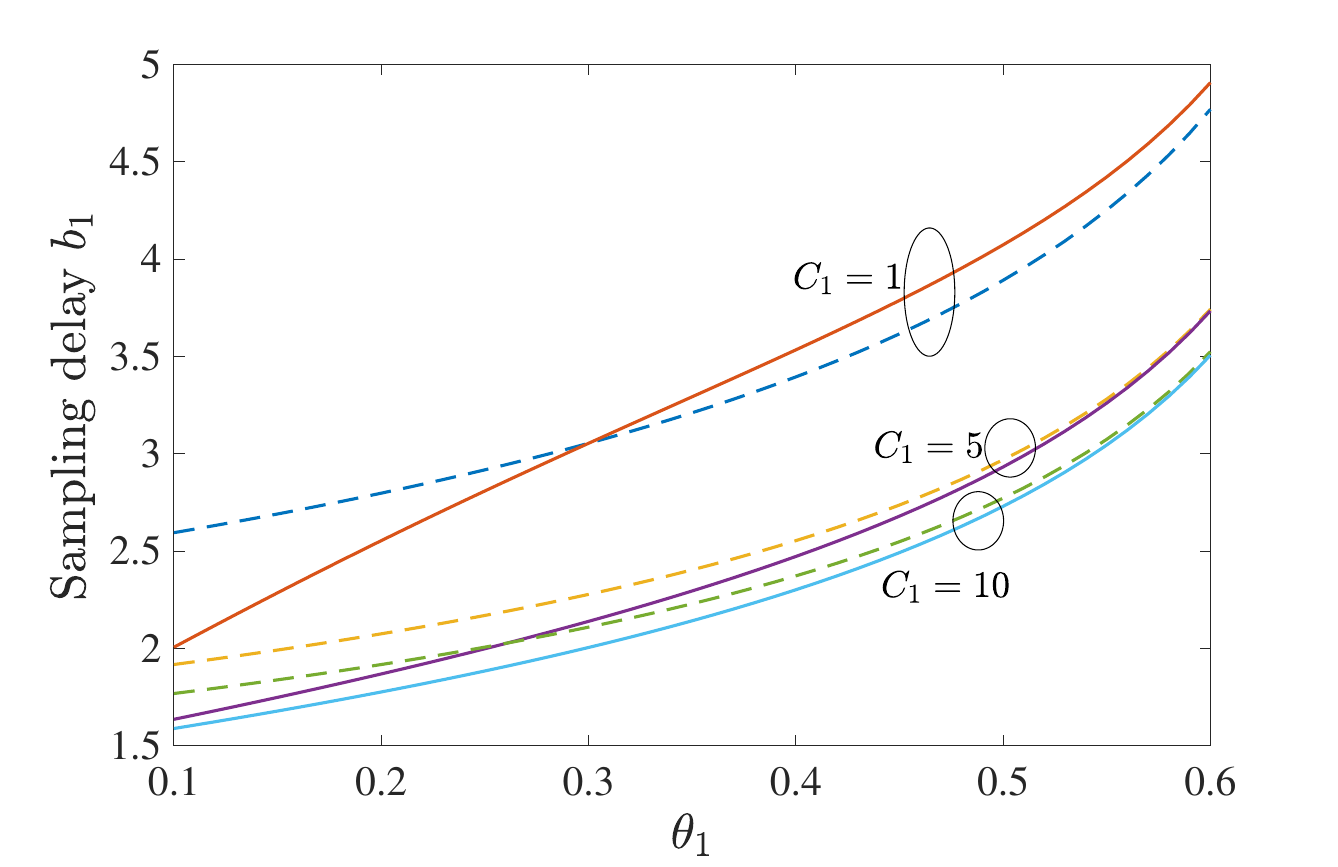}
	\caption{Sampling delay $b_1$ versus outage exponent $\theta_1$.}
	\label{fig:sim1}
\end{figure}
\begin{figure}
\centering
	\includegraphics[width=.4\textwidth]{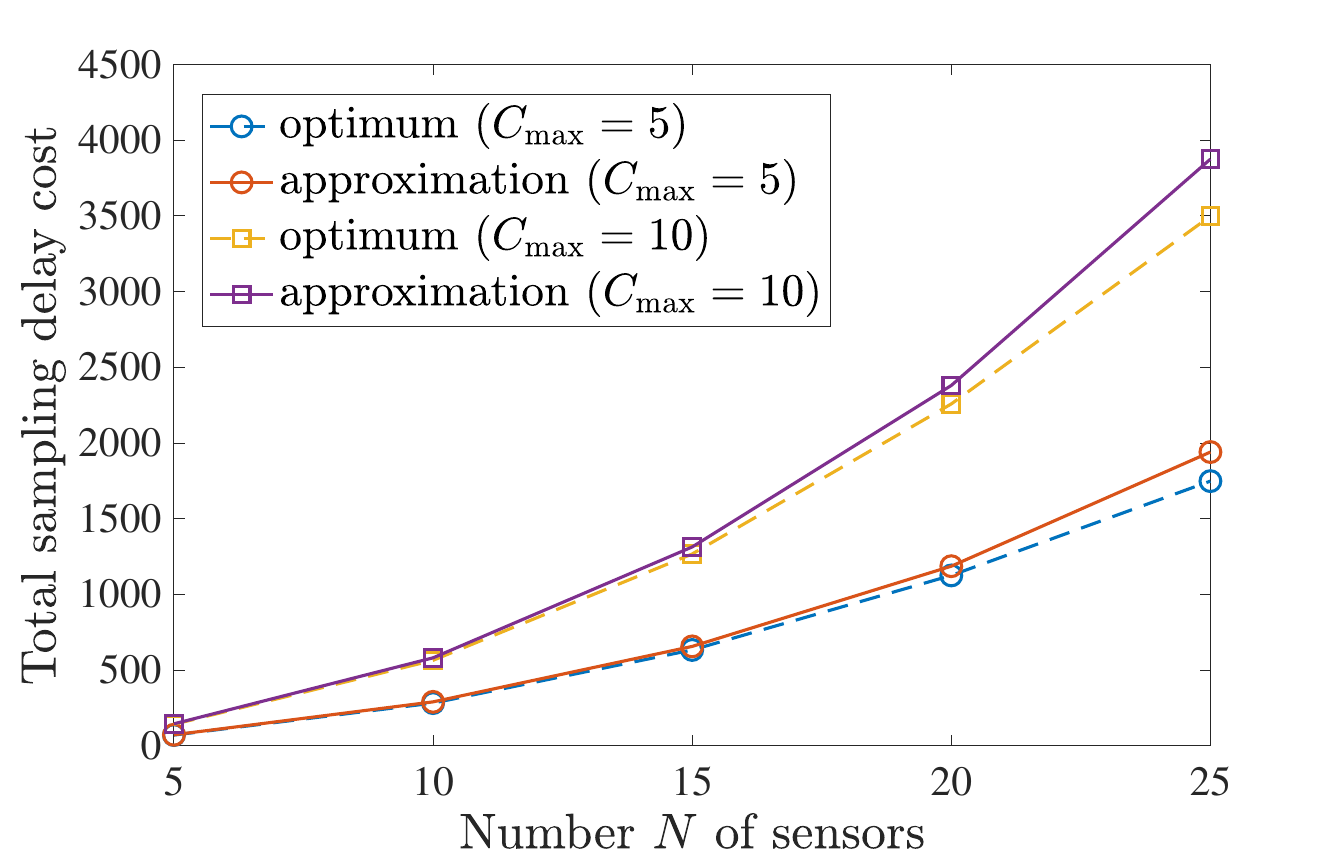}
	\caption{The total sampling delay cost versus the number $N$ of sensors.}
	\label{fig:sim2}
\end{figure}
In this section, we evaluate the proposed approximate solution through numerical studies. For simplicity, we set $\mu_i = 1$ for all $i$. 

First, we consider a system with two sensors, where $\theta_2 = 0.3$ and $C_2 = 1$ are fixed. Fig.~\ref{fig:sim1} illustrates the trade-off between  sampling delay $b_1$ and  outage exponent $\theta_1$ under the optimal resource allocation (dashed curves) and the approximate solution (solid curves). From Fig.~\ref{fig:sim1}, we observe that the approximate trade-off closely matches the optimal trade-off when $\theta_1$ approaches its maximum feasible value (i.e., when $\boldsymbol{\theta}$ is near the boundary of the feasibility region) or when the sampling delay cost $C_1$ becomes larger.

Second, we examine the accuracy of the approximation in terms of the total sampling delay cost (indicative of a weighted sum of throughput), as shown in Fig.~\ref{fig:sim2}. Here, the outage exponents are set as follows: $\theta_{\frac{N}{2}} = \frac{0.5}{N}$, $\theta_i = \frac{0.5}{N} - 0.01(\frac{N}{2} - i)$ for $i < \frac{N}{2}$, and $\theta_i = \frac{0.5}{N} + 0.01(i - \frac{N}{2})$ for $i > \frac{N}{2}$. This configuration ensures $\sum_{i=1}^N \theta_i = 0.5$, which is far from the boundary of the feasibility region. The sampling delay cost $C_i$ is randomly selected from the range $[1, C_{\max}]$ for all $i$, and each data point in the figure is averaged over 1000 random scenarios. From Fig.~\ref{fig:sim2}, we observe that the approximation performs closely to the optimal solution (in terms of total sampling delay cost), even for small values of $N$.

\section{Conclusion}
In this paper, we investigated an information collection system operating over a shared communication resource. We proposed a joint design of sampling  and resource allocation to approximately minimize the total sampling delay cost while satisfying PAoI outage exponent constraints. While this work focuses on a feasible set of outage exponents, an interesting direction for future research is the development of admission control mechanisms for systems where the outage exponents exceed the feasibility region.

 \appendices
 \section{Proof of Lemma~\ref{lemma:necessary-condition}}\label{appendix:lemma:necessary-condition}
 The proof considers a fixed sensor $i$. From \cite{seo2019outage}, PAoI can   be expressed by $A_{i, j}=\max_{1 \leq s \leq j}\{\sum_{k=s}^{j} T_{i, k}-(j-1-s) b\},$ which can be expanded as $\max\{T_{i, j}+b, T_{i, j-1}+T_{i, j}, T_{i, j-2}+T_{i, j-1}+T_{i, j}-b, \cdots\}$. 	
Because of i.i.d. transmission times $T_{i,j}$ for different $j$, we can re-express  PAoI as $A_{i, j}=\max_{1 \leq s \leq j}\left\{\sum_{k=1}^{s} T_{i, k}-s b+2 b\right\}$. In the steady state, as $j \to \infty$, this becomes $A_{i,\infty} = \max_{s \geq 1} \left\{ \sum_{k=1}^{s} T_{i,k} - s b + 2b \right\}$.

Let $I_i(x) = \sup_{\gamma \in \mathbf{R}} \{\gamma x - \Lambda_i(\gamma)\}$ be the rate function associated with  transmission time $T_i$. We first derive an upper bound on the outage probability:
\begin{align*}
\mathbb{P}[A_{i,\infty} \geq x] &= \mathbb{P}\left[\max_{s \geq 1} \left\{ \sum_{k=1}^{s} T_{i,k} - s b + 2b \right\} \geq x\right] \\
&\leq \sum_{s=1}^{\infty} \mathbb{P}\left[\sum_{k=1}^{s} T_{i,k} - s b + 2b \geq x\right] \\
&= \sum_{s=1}^{\infty} \mathbb{P}\left[\sum_{k=1}^{s} T_{i,k} \geq x + s b - 2b\right] \\
&\stackrel{(a)}{\leq} \sum_{s=1}^{\infty} \frac{\mathbb{E}\left[e^{\gamma \sum_{k=1}^{s} T_{i,k}}\right]}{e^{\gamma(x + s b - 2b)}} \text{ for any $\gamma > 0$}\\
&\stackrel{(b)}{=} \sum_{s=1}^{\infty} \frac{e^{s \Lambda_i(\gamma)}}{e^{\gamma(x + s b - 2b)}} \\
&= e^{2\gamma b} \sum_{s=1}^{\infty} e^{-s\left(\gamma \left(\frac{x}{s} + b \right) - \Lambda_i(\gamma)\right)}\\
&\mathop{=}^{(c)} e^{2\gamma b}\sum_{t \in \{\frac{1}{x}, \frac{2}{x}, \cdots\}} e^{-x t \left(\gamma \left(\frac{1}{t} + b \right) - \Lambda_i(\gamma)\right)},
\end{align*}
where step (a) applies the Chernoff bound, step (b) uses the i.i.d. assumption for $T_{i,k}$'s, and step (c) defines another variable $t = \frac{s}{x}$. Note that the above inequality holds for all $\gamma > 0$. We can choose $\gamma^* = \arg\sup_{\gamma > 0} \gamma \left(\frac{1}{s} + b\right) - \Lambda_i(\gamma)$. Note that $\gamma^*$ is finite because, according to \cite[Lemma 1]{seo2019outage}, if $\gamma^*$ were infinite, the outage probability of  PAoI would be zero, which is not feasible in practice. Substituting $\gamma^*$ into the inequality, we can obtain:
\begin{align*}
\mathbb{P}[A_{i,\infty} \geq x] &\leq e^{2\gamma^* b} \sum_{t \in \{\frac{1}{x}, \frac{2}{x}, \cdots\}} e^{-x t \left(\gamma^* \left(\frac{1}{t} + b\right) - \Lambda_i(\gamma^*)\right)} \\	
&\stackrel{(a)}{=} e^{2\gamma^* b} \sum_{t \in \{\frac{1}{x}, \frac{2}{x}, \cdots\}} e^{-x t I_i\left(\frac{1}{t} + b\right)} \\
&\stackrel{(b)}{\leq} e^{2\gamma^* b} \hat{t} x e^{-x \inf_{t>0} t I_i\left(\frac{1}{t} + b\right)} \text{ for some constant $\hat{t}$}, 
\end{align*}
where step (a) applies the definition of the rate function $I_i(x)$ and the equality $\sup_{\gamma \in \mathbf{R}} \gamma x - \Lambda_i(\gamma) = \sup_{\gamma > 0} \gamma x - \Lambda_i(\gamma)$ from \cite[Page 303]{srikant2014communication}, and step (b)  bounds the sum of terms beyond $\hat{t}x-1$ (which is a geometric series) by the maximum of the first $\hat{t}x - 1$ terms, following the bounding technique in \cite[Page 315]{srikant2014communication}.

Next, we derive a lower bound for the outage probability:
\begin{align*}
\mathbb{P}[A_{i,\infty} \geq x] &= \mathbb{P}\left[\max_{s \geq 1} \left\{\sum_{k=1}^{s} T_{i,k} - s b + 2b \right\} \geq x\right] \\
&\geq \mathbb{P}\left[\sum_{k=1}^{s} T_{i,k} - s b + 2b \geq x\right]  \text{ for any $s$}  \\
&\geq \mathbb{P}\left[\sum_{k=1}^{s}T_{i,k} - s b \geq x\right] \\
&\stackrel{(a)}{\geq} e^{-x \inf_{t>0} \frac{1}{t} I_i(t + b)},
\end{align*}
where step (a) is from the Cramer–Chernoff theorem as in \cite[Page 309]{srikant2014communication}.

Combining these bounds, we can express the exponent of the outage probability as
\[
\lim_{x \to \infty} \frac{\ln \mathbb{P}[A_{i,\infty} \geq x] }{x} = -\inf_{t>0} \frac{1}{t} I_i(t + b).
\]
Finally, we can establish the equivalence as follows:
\begin{align*}
&\lim_{x\rightarrow \infty}\frac{\ln\mathbb{P}[A_{i,\infty} \geq x]}{x} \leq -\theta_i\\
\iff &  \inf_{t > 0} \frac{1}{t} I_{i}(t+b) \geq \theta_{i} \\
\iff &  \frac{1}{t} I_{i}\left(t+b\right) \geq \theta_{i} \text { for all $t>0$}  \\
\mathop{\iff}^{(a)} &   t\theta_{i}-I_{i}(t+b) \leq 0  \text { for all $t>0$ and also $t\leq 0$}  \\
\iff &  \theta_{i}(t+b)-I_{i}(t+b) \leq \theta_{i} b  \text { for all $t\in \mathbf{R}$}\\
\iff &  \max _{t+b \in \mathbf{R}} \theta_{i}(t+b)-I_{i}(t+b) \leq \theta_{i}b \\
\mathop{\iff}^{(b)} & \Lambda_i(\theta_{i}) \leq \theta_{i} b,
\end{align*}
where step (a) is because $I(x) \geq 0$ for all $x \in \mathbf{R}$, step (b)  uses the property  between the rate function $I_i(x)$ and the LMGF $\Lambda_i(\gamma)$, as in \cite[Theorem 10.3.4]{srikant2014communication}.

\clearpage
\small
\bibliographystyle{IEEEtran}
\bibliography{IEEEabrv,ref}
\balance

\end{document}